\documentclass[letterpaper,10 pt,conference]{ieeeconf}  

\IEEEoverridecommandlockouts                              
\usepackage{amsmath,amsthm}
\usepackage{amsfonts}
\usepackage{dsfont}
\usepackage{graphicx}
\usepackage{color}
\usepackage{cite}
\usepackage{subcaption}
\usepackage{flushend}
\usepackage{graphicx}
\usepackage{bm}
\usepackage{stmaryrd}
\newtheorem{theorem}{Theorem}
\newtheorem{proposition}{Proposition}
\newtheorem{property}{Property}
\newtheorem{lemma}{Lemma}

\newtheorem{definition}{Definition}

\newtheorem{remark}{Remark}
\newcommand{\mathbd}[1]{\mbox{\boldmath $#1$}}

\DeclareMathOperator*{\argmin}{arg\,min}

\def\bmP{\mbox{$\bm{P}$}}
\def\bmC{\mbox{$\bm{C}$}}
\def\bmF{\mbox{$\bm{F}$}}
\def\bmD{\mbox{$\bm{\Delta}$}}
\def\bmT{\mbox{$\bm{T}$}}
\def\bmI{\mbox{$\bm{I}$}}

\title{\LARGE \bf
Stabilization of Cascaded Two-Port Networked Systems Against Nonlinear Perturbations
}

\author{Di Zhao, Sei Zhen Khong, and Li Qiu
\thanks{*This work was supported in part by the Research Grants Council of Hong Kong Special Administrative Region, China, under projects \textcolor{black}{16201115} and T23-701/14N. }
\thanks{D. Zhao and L. Qiu are with the Department of Electronic \& Computer Engineering, The Hong Kong University of Science and Technology, Clear Water Bay, Kowloon, Hong Kong, China. {\tt\small dzhaoaa@ust.hk}, {\tt\small eeqiu@ece.ust.hk}}
\thanks{S. Z. Khong is with the Institute for Mathematics and its Applications, University of Minnesota, Minneapolis, MN 55455, USA. {\tt\small szkhong@umn.edu}}
}
\graphicspath{{./fig/}}

\begin{document}

\maketitle
\thispagestyle{empty}
\pagestyle{empty}

\begin{abstract}

  A networked control system (NCS) consisting of cascaded two-port communication channels between the plant and controller is modeled and
  analyzed. Towards this end, the robust stability of a standard closed-loop system in the presence of conelike perturbations on the system graphs is
  investigated. The underlying geometric insights are then exploited to analyze the two-port NCS. It is shown that the robust stability of the
  two-port NCS can be guaranteed when the nonlinear uncertainties in the transmission matrices are sufficiently small in norm. The stability
  condition, given in the form of ``arcsin'' of the uncertainty bounds, is both necessary and sufficient.

\end{abstract}

\section{Introduction}
Feedback is widely used for handling modeling uncertainties in the area of systems and control. Within a feedback loop, communication between the
plant and controller plays an important role in that the achieved control performance and robustness heavily rely on the quality of communication. In
practice, communication can never be ideal due to the presence of channel distortions and interferences. In this study, we analyze the robust
stability of a feedback system involving bidirectional uncertain communication modeled by cascaded two-port networks.

Most control systems can be regarded as structured networks with signals \textcolor{black}{transmitted} through channels powered by various devices, such as sensors or
satellites. A networked control system (NCS) differs from a standard closed-loop system in that the information is exchanged through a communication
network \cite{zhang2001stability}. The presence of such a network may introduce disturbances to a control system and hence significantly compromise its
performance.

In this study, we introduce an NCS model, extending the standard linear time-invariant (LTI) closed-loop system (Fig.~\ref{figloop}) to the feedback
system with cascaded two-port connections (Fig.~\ref{fignetwork}). We assume that the controller and plant are LTI while the two-port networks involve
nonlinear perturbations on their transmission matrices. In terms of communication uncertainties, we model the transmission matrices as
$\bmT = \bmI + \bmD$, where $\bmD$ is a bounded nonlinear operator. Our formulation of robust stabilization problem is mainly motivated by the
application scenario of stabilizing a feedback system where the plant and controller do not possess an ideal communication environment and their input-output
signals can only be sent through communication networks with several relays, as in, for example, teleoperation systems\cite{anderson1989bilateral},
satellite networks \cite{Alagoz2011POI}, wireless sensor networks \cite{Kumar2014CST} and so on. Moreover, each sub-system between two neighbouring
relays, representing a communication channel, may involve not only multiplicative distortions on the transmitted signal itself but also additive
interferences caused by the signal in the reverse direction, which is usually encountered in a bidirectional wireless network subject to channel fading
or under malicious attacks \cite{wu2007survey}.

Two-port networks are not a new concept and have been studied for decades for different purposes. \textcolor{black}{Historically, two-port networks were
  first introduced in electrical circuits theory \cite{bakshi2009network}.} Later on they were utilized to represent LTI systems in the so-called
chain-scattering formalism \cite{kimura1996chain}, which is essentially a two-port
network. 
Such representations have also been used for studying feedback robustness from the perspective of the $\nu$-gap metric~\cite{Khong2013TAC}.  Recently,
approaches based on the two-port network to modeling communication channels in a networked feedback system is studied in \cite{gu2011cdc} and
\cite{di2016cdc}. There, uncertain two-port connections are used to introduce channel uncertainties, based on which we propose our cascaded two-port
communication model with nonlinear perturbations in this paper.

One of the contributions of our study is a clean result for analyzing the stability of a feedback system with multiple sources of uncertainties. A
general approach to robust stabilization of LTI systems with structured uncertainties is $\mu$ analysis, \textcolor{black}{which is known to be
  computationally intractable in general in the presence of multiple uncertainties \cite{zhou1998essentials}}. Furthermore, the two-port uncertainties
in this study are nonlinear, which bring in an additional obstacle. To overcome these difficulties, we take advantage of the special two-port
structures and make use of geometric insights on system stability via an input-output approach. By generalizing the ``arcsin'' theorem in
\cite{qiu1992feedback} for a standard closed-loop system, we are able to give a concise necessary and sufficient robust stability condition for the
two-port NCS. \textcolor{black}{Moreover, the stability condition is scalable and computationally friendly, in the sense that when the topology of the
  two-port NCS is changed, the stability condition can be efficiently updated based only on the modified components.  In terms of designing an optimal
  controller, it suffices to solve an $\mathcal{H}_\infty$ optimization problem, which is mathematically tractable.}

\begin{figure}
  \centering
  \includegraphics[scale=0.6]{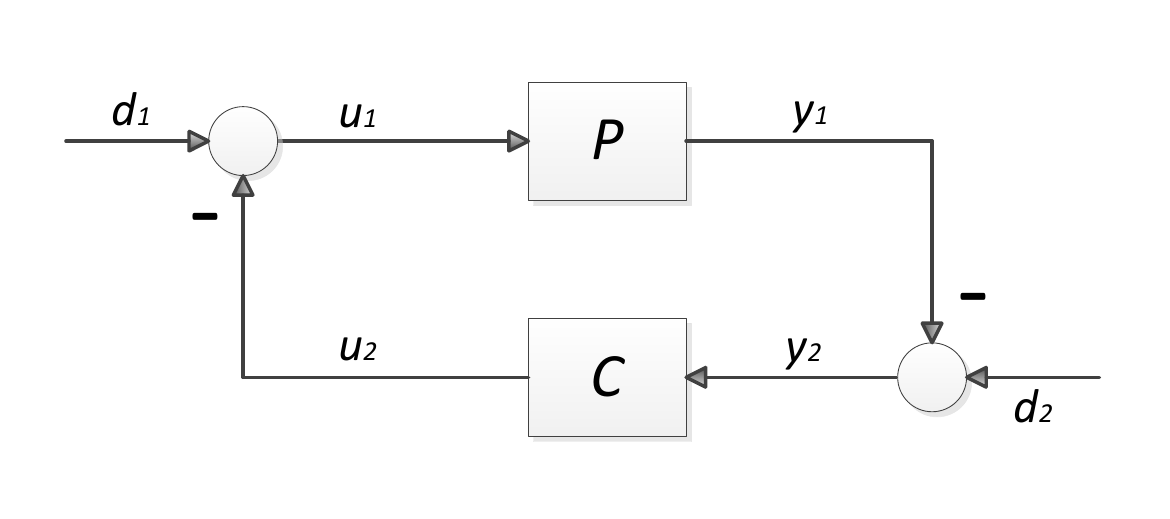}\\
  \caption{A Standard Closed-Loop System}\label{figloop}
\end{figure}

\begin{figure}
  \centering
  \includegraphics[scale=0.58]{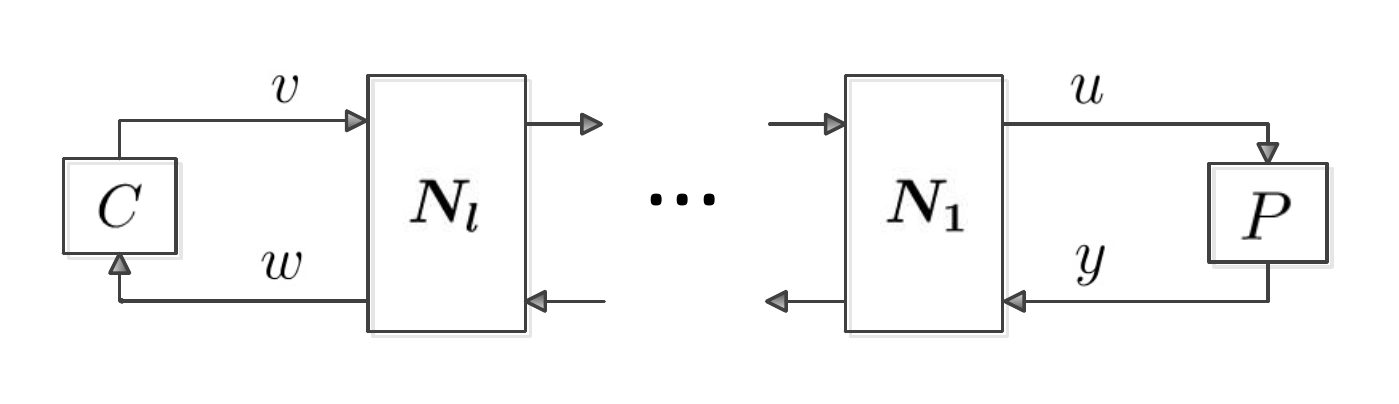}\\
  \caption{Communication Channels Modeled by Cascaded Two-port Networks}\label{fignetwork}
\end{figure}

It is worth noting that there exist previous works on robust stabilization of NCSs with special architectures and various uncertainty
descriptions. For example, \cite{anderson1989bilateral} considers teleoperation of robots through two-port communication networks with time-delay,
\cite{Ishii2010cdc} considers a plant with parametric uncertainties over networks subject to packet loss, \cite{Chesi2016TAC} considers a plant with
polytopic uncertainties in its coefficients over a communication channel subject to fadings and so on. The differences of our work from the previous
ones are that our channel model characterizes bi-directional communication involving both distortions and interferences and these uncertainties may be nonlinear.

The rest of the paper is organized as follows. First in Section \ref{sec:pre}, we define open-loop stability, closed-loop well-posedness and
stability, system uncertainties and some related properties. Then in Section \ref{sec:Closeloop}, we give a robust stability condition for a
closed-loop system with conelike uncertainty descriptions. Thereafter in Section \ref{sec:network}, we extend the results on robust stability to
cascaded two-port networks. In Section \ref{sec:conc}, we conclude this study and summarize our contributions.
%

\section{Preliminaries}\label{sec:pre}
\subsection{Open-loop Stability}


Let $\mathcal{H}_2^n
:= \{f : [0, \infty) \to \mathbb{R}^n\ |\ \|f\|_2^2 := \int_0^\infty |f(t)|^2 \, dt < \infty\}$, where
$|\cdot|$
denotes the Euclidean norm. Let $\mathcal{RH}_\infty$
consist of all the real rational members of $\mathcal{H}_\infty$,
the Hardy $\infty$-space of functions that are holomorphic on the right-half complex plane.

Denote the time truncation operator at time $\tau\in [0,\infty)$ as $\bm{T}_\tau$, such that for $u(t) \in \mathcal{H}_2$, $$(\bm{T}_\tau u)(t) = \left\{
\begin{array}{ll}
  u(t), & \hbox{$0\leq t<\tau$;} \\
  0, & \hbox{Otherwise.}
\end{array}
\right.
$$

A nonlinear system is represented by an operator $\bm{P}: \text{dom}(\bmP) \subset \mathcal{H}_2 \mapsto \mathcal{H}_2$ with domain
$\text{dom}(\bmP) = \{u \in \mathcal{H}_2\ |\ \bmP u \in \mathcal{H}_2\}$. We denote its image as $\text{img}(\bmP)$. A physical system should
additionally be causal, which is defined as follows \cite{Willems1971nonlinear}.

\begin{definition}
A nonlinear system $\bm{P}: \emph{dom}(\bmP) \subset \mathcal{H}_2 \mapsto  \mathcal{H}_2$ is said to be causal if for every $\tau \in [0,\infty)$ and $u_1,u_2 \in \emph{dom}(\bmP)$,
$$\bm{T}_\tau u_1 = \bm{T}_\tau u_2 \Rightarrow \bm{T}_\tau \bmP u_1 = \bm{T}_\tau \bmP u_2$$
\end{definition}
We assume $\bmP 0 = 0$ throughout this study, which means every nonlinear system we consider has zero output whenever the input is zero. The finite-gain
stability of a system is defined as follows \cite{Vidyasagar1993nonlinear}.

\begin{definition}
  A causal nonlinear operator (system) $\bmP$ is said to be (finite-gain) stable if $\emph{dom}(\bmP)=\mathcal{H}_2$ and its operator norm is bounded,
  that is
\begin{align*}
\|\bm{P}\| := \sup_{0 \neq x\in \mathcal{H}_2}\frac{\|\bm{P}x\|_2}{\|x\|_2} < \infty.
\end{align*}
\end{definition}

\subsection{Closed-loop Stability}
We consider a standard closed-loop system in Fig.~\ref{figloop} with plant $\bmP: \text{dom}(\bmP) \subset \mathcal{H}^p_2 \mapsto \mathcal{H}^m_2$
and controller $\bmC: \text{dom}(\bmC) \subset \mathcal{H}^m_2\mapsto \mathcal{H}^p_2$. In the following, the superscripts of $\mathcal{H}^m_2$ and
$\mathcal{H}^p_2$ will be omitted for notational simplicity.

The graph of $\bmP$ is defined as
$$\mathcal{G}_{\bm{P}} = \begin{bmatrix}
\bmI \\
\bmP \\
\end{bmatrix}\text{dom}(\bmP)$$
and similarly the inverse graph of $\bmC$ is defined as
$$\mathcal{G}'_{\bm{C}} = \begin{bmatrix}
\bmC \\
\bmI \\
\end{bmatrix}\text{dom}(\bmC),$$
both of which are assumed to be closed in this study.

It can be seen in \cite{Willems1971nonlinear,Vidyasagar1993nonlinear,SeiZhen_AUCC13} that various versions of feedback well-posedness may be assumed
based on different signal spaces and causality requirements. In this study, we adopt the well-posedness definition from \cite{SeiZhen_AUCC13} without
appealing to extended spaces, by contrast to, for example, \cite{Willems1971nonlinear,Vidyasagar1993nonlinear}.
\begin{definition}\label{def:wellposed}
The closed-loop system $[\bmP,\bmC]$ is said to be well-posed if
\begin{align*}
  \bmF_{\bm{P},\bm{C}}:&~\emph{dom}(\bmP)\times \emph{dom}(\bmC) \mapsto \mathcal{H}_2\\
  &:= \begin{bmatrix}
                u_1 \\
                y_2 \\
              \end{bmatrix} \mapsto \begin{bmatrix}
                d_1 \\
                d_2 \\
              \end{bmatrix} = \begin{bmatrix}
                \bmI & \bmC\\
                \bmP & \bmI \\
              \end{bmatrix}\begin{bmatrix}
                u_1 \\
                y_2 \\
              \end{bmatrix}
\end{align*}
is causally invertible on $\emph{img}(\bmF_{\bm{P},\bm{C}})$.
\end{definition}

Correspondingly, the stability of the closed-loop system is defined as follows:
\begin{definition}\label{def:closed_stable}
A well-posed closed-loop system $[\bmP,\bmC]$ is (finite-gain) stable if
$\bmF_{\bm{P},\bm{C}}$ is surjective and $\bmF^{-1}_{\bm{P},\bm{C}}$ is finite-gain stable.
\end{definition}

When $\bmF_{\bm{P},\bm{C}}$ is surjective, the parallel projection operators \cite{DOYLE199379} along $\mathcal{G}_{\bm{P}}$ and
$\mathcal{G}'_{\bm{C}}$, $\Pi_{\mathcal{G}_{\bm{P}}\sslash\mathcal{G}'_{\bm{C}}}$ and $\Pi_{\mathcal{G}'_{\bm{C}}\sslash\mathcal{G}_{\bm{P}}}$, can be defined
respectively as
\begin{equation}\label{eq:pi_pc}
\begin{aligned}
\Pi_{\mathcal{G}_{\bm{P}}\sslash\mathcal{G}'_{\bm{C}}}&: \begin{bmatrix}
                                                    d_1 \\
                                                    d_2 \\
\end{bmatrix}\in \mathcal{H}_2 \mapsto
\begin{bmatrix}
u_1 \\
y_1 \\
\end{bmatrix}\in \mathcal{G}_{\bm{P}}\\
& = \begin{bmatrix}
      \bmI & 0 \\
      0 & -\bmI \\
    \end{bmatrix}\bmF^{-1}_{\bm{P},\bm{C}}+\begin{bmatrix}
                                             0 & 0 \\
                                             0 & \bmI \\
                                           \end{bmatrix},
\end{aligned}
\end{equation}
\begin{equation}\label{eq:pi_cp}
\begin{aligned}
\Pi_{\mathcal{G}'_{\bm{C}}\sslash\mathcal{G}_{\bm{P}}}&: \begin{bmatrix}
                                                    d_1 \\
                                                    d_2 \\
\end{bmatrix}\in \mathcal{H}_2 \mapsto
\begin{bmatrix}
u_2 \\
y_2 \\
\end{bmatrix}\in \mathcal{G}'_{\bm{C}}\\
& = \begin{bmatrix}
      -\bmI & 0 \\
      0 &\bmI \\
    \end{bmatrix}\bmF^{-1}_{\bm{P},\bm{C}}+\begin{bmatrix}
                                             \bmI & 0 \\
                                             0 & 0 \\
                                           \end{bmatrix}.
\end{aligned}
\end{equation}
It follows that every $w \in \mathcal{H}_2$ has a unique decomposition as $w=m+n$ with $m=\Pi_{\mathcal{G}_{\bm{P}}\sslash\mathcal{G}'_{\bm{C}}}w \in \mathcal{G}_{\bm{P}}$ and $n=\Pi_{\mathcal{G}'_{\bm{C}}\sslash\mathcal{G}_{\bm{P}}}w \in \mathcal{G}'_{\bm{C}}$.

The next proposition bridges the finite-gain stability and the boundedness of parallel projections \cite{DOYLE199379}.
\begin{proposition}\label{prop:pi_pc_stability}
A well-posed closed-loop system $[\bmP,\bmC]$ is stable if and only if $\bmF_{\bm{P},\bm{C}}$ is surjective and $\Pi_{\mathcal{G}_{\bm{P}}\sslash\mathcal{G}'_{\bm{C}}}$ or $\Pi_{\mathcal{G}'_{\bm{C}}\sslash\mathcal{G}_{\bm{P}}}$ is finite-gain stable.
\end{proposition}

For a finite-gain stable closed-loop system $[\bmP,\bmC]$, its stability margin is defined as $b_{\bm{P},\bm{C}} :=
\|\Pi_{\mathcal{G}_{\bm{P}}\sslash\mathcal{G}'_{\bm{C}}}\|^{-1}$. It is shown in \cite{DOYLE199379} that if either $\bm{P}$ or $\bm{C}$ is linear, then
$b_{\bm{P},\bm{C}} = b_{\bm{C},\bm{P}}$.
%

\subsection{System Uncertainties}

%
%

A well-known method to introduce system uncertainties is through various variants of the ``gap'' or ``aperture'' between system graphs
\cite{vinnicombe2000uncertainty}. In this study, before characterizing the uncertainties in two-port networks, we introduce a useful notion of
neighborhood of a certain nominal system's graph, which may serve as its uncertainty set. Let $\mathcal{M}$ be a manifold in $\mathcal{H}_2$. Define
the conelike neighborhood of $\mathcal{M}$ as
$$\mathcal{S}(\mathcal{M},r) = \{n\in \mathcal{H}_2 : \inf_{0 \neq m\in \mathcal{M}} \frac{\|n-m\|_2}{\|m\|_2}\leq r\}\cup \{0\}.$$

\begin{figure}
  \centering
  \includegraphics[width=.4\textwidth]{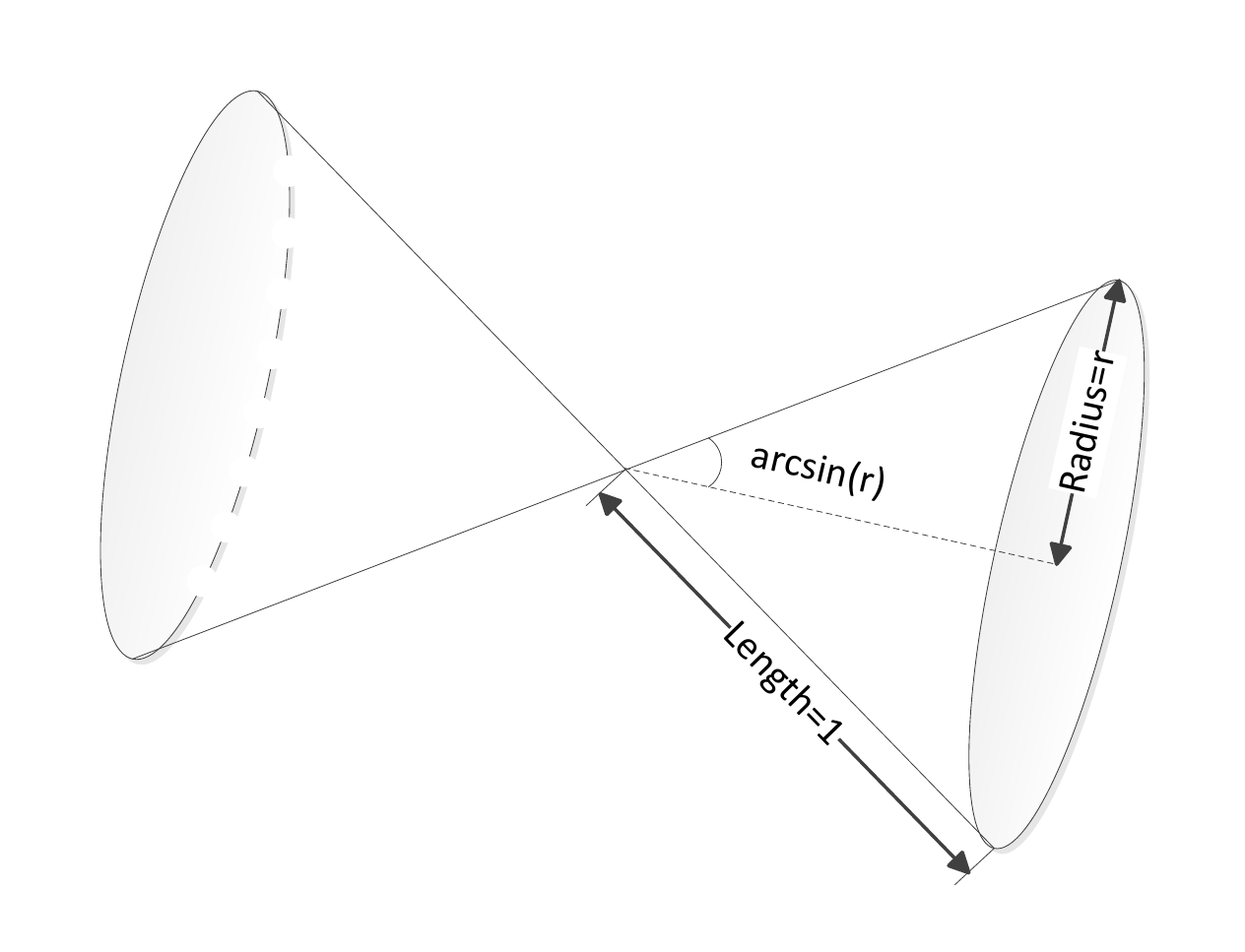}\\
  \caption{The conelike neighborhood in $\mathbb{R}^3$}\label{fig:double_cone}
\end{figure}

If $\mathcal{M}$ is a one-dimensional subspace in $\mathbb{R}^3$, the set $\mathcal{S}(\mathcal{M},r)$ is simply a right circular double cone as shown in Fig.~\ref{fig:double_cone}. In case of one-dimensional subspace in $\mathbb{R}^2$, the set can be interpreted as \textcolor{black}{doubly sector-bounded area \cite{zames1966sector}.} If $\mathcal{M}$ is the graph of certain linear system, it ``resembles'' a closed double cone in the space of $\mathcal{H}_2$, which provides us some geometric intuitions on the system uncertainties.

Based on the Hilbert space structure of $\mathcal{H}_2$, let $\theta(x,y)$ denote the acute angle between
$x,y \in \mathcal{H}_2\setminus \{0\}$ and $\theta(x,y)=\infty$ if either of $x,y$ is zero almost everywhere for convenience.

%

\textcolor{black}{Given $r \in (0,1)$ and a closed conelike neighboring set $\mathcal{M} \subset \mathcal{H}_2$, we have the following useful properties:}
\begin{property}\label{prope1}
\textcolor{black}{Let $ n \in \mathcal{H}_2\setminus \{0\}$. Then $n \in \mathcal{S}(\mathcal{M},r)$ if and only if $\alpha n \in \mathcal{S}(\mathcal{M},r)$ for every $\alpha \in \mathbb{R};$}
\end{property}
\begin{property}\label{prope2}
  $\mathcal{S}(\mathcal{M},r) = \{n\in \mathcal{H}_2 : \inf_{m\in \mathcal{M}} \theta(m,n)\leq \arcsin r\}\cup \{0\}$.
\end{property}
Another related neighboring set is defined as follows:
$$\textcolor{black}{\tilde{\mathcal{S}}(\mathcal{M},r):= \{n\in \mathcal{H}_2\setminus \{0\} : \inf_{m\in \mathcal{M}} \frac{\|n-m\|_2}{\|n\|_2}\leq r\}\cup \{0\}.}$$
\begin{property}\label{prope3}
$\textcolor{black}{\mathcal{S}(\mathcal{M},r) = \tilde{\mathcal{S}}(\mathcal{M},r).}$
\end{property}
\textcolor{black}{The proofs for the above properties are in Appendix \ref{app:pf_properties}, which follow from the definition of conelike
  neighborhoods directly.}
\begin{remark}
  \textcolor{black}{In general, $\tilde{\mathcal{S}}(\mathcal{M},r) \neq {\mathcal{S}}(\mathcal{M},r)$ for arbitrary manifold $\mathcal{M}$ in
    $\mathcal{H}_2$.}
\end{remark}
One benefit of defining uncertainties as above is that we can examine the intersection of two cones simply by studying the angles between two lines
from each of them respectively. Moreover, the intersection of the graphs may reflect the instability of a closed-loop system, as is detailed in the next section.

\section{Feedback interconnections with Conelike Uncertainties}\label{sec:Closeloop}
%

%
%

Given a (possibly unstable) LTI nominal closed-loop system $[P,C]$ with open-loop system graphs $\mathcal{G}_{{P}}$ and $\mathcal{G}'_{{C}}$, we have
the following result concerning its robust stability, whose proof is in Appendix \ref{app:pf_propo2}.

\begin{proposition}\label{prop:nonlinear_stability_criterion}
  Given $r_p,r_c \in (0,1)$, the perturbed system $[\bmP_1,\bmC_1]$ is stable for all
  $\mathcal{G}_{\bm{P}_1}\subset \mathcal{S}(\mathcal{G}_{{P}},r_p),~\mathcal{G}'_{\bm{C}_1}\subset \mathcal{S}(\mathcal{G}'_{{C}},r_c)$ such that $\bmF_{\bm{P}_1,\bm{C}_1}$ is surjective
  if and only if $$\mathcal{S}(\mathcal{G}_{{P}},r_p)\cap \mathcal{S}(\mathcal{G}'_{{C}},r_c) = \{0\}.$$
\end{proposition}
It is known that a standard well-posed LTI closed-loop system $[P,C]$ is stable if and only if
$\mathcal{G}_P \oplus \mathcal{G}'_C = \mathcal{H}_2 \times \mathcal{H}_2$. As there is no subspace representation for the graph of a nonlinear
system, Proposition \ref{prop:nonlinear_stability_criterion} generalizes the geometric insight of complementarity of subspaces. Building on that, we
have the following robust stability condition, which extends the ``arcsin'' inequality condition in \cite{qiu1992feedback} and
\cite{qiu1992pointwise}.
\begin{theorem}\label{thm:extend_arcsin}
Assume the LTI nominal closed-loop system $[P,C]$ is stable. The following statements are equivalent:

 \begin{enumerate}
 \item The perturbed system $[\bmP_1,\bmC_1]$ is stable for all $\mathcal{G}_{\bm{P}_1} \subset \mathcal{S}(\mathcal{G}_{{P}},r_p)$,
   $\mathcal{G}'_{\bm{C}_1} \subset \mathcal{S}(\mathcal{G}'_{{C}},r_c)$ such that $\bmF_{\bm{P}_1,\bm{C}_1}$ is surjective;
   \item $\mathcal{S}(\mathcal{G}_{{P}},r_p)\cap \mathcal{S}(\mathcal{G}'_{{C}},r_c) = \{0\};$
   \item $\arcsin r_p +\arcsin r_c < \arcsin b_{{P},{C}}.$
 \end{enumerate}

\end{theorem}
\begin{proof}
  The equivalence between 1) and 2) has been established in Proposition \ref{prop:nonlinear_stability_criterion}.  The direction $1)\Rightarrow 3)$
  follows from the ``arcsin'' theorem in \cite{qiu1992feedback} for LTI systems by noting that standard the gap metric balls
  $\mathcal{B}(\mathcal{G}_{{P}},r_p)$ and $\mathcal{B}(\mathcal{G}_{{P}},r_c)$ are contained in the conelike sets $\mathcal{S}(\mathcal{G}_{{P}},r_p)$
  and $\mathcal{S}(\mathcal{G}'_{{C}},r_c)$, respectively.

  Next we show $3)\Rightarrow 2)$. First note that $b_{{P},{C}} = \inf_{m\in \mathcal{G}_{{P}},n\in \mathcal{G}'_{{C}}} \sin \theta(m,n)$ from
  \cite{qiu1992feedback}. Given any $m_1 \in \mathcal{S}(\mathcal{G}_{{P}},r_p)$ and $n_1 \in \mathcal{S}(\mathcal{G}'_{{C}},r_c)$, the triangle
  inequality for $\theta(\cdot,\cdot)$ gives
\begin{align}\label{ineq:angles1}
\theta(m_1,n_1) \geq \theta(m,n)-\theta(m,m_1)-\theta(n,n_1)
\end{align}
for any $m \in \mathcal{G}_{{P}}$ and $n \in \mathcal{G}'_{{C}}$.

It follows directly that $\inf_{m \in \mathcal{G}_{\bm{P}}} \theta (m,m_1) \leq \arcsin r_p$ and $\inf_{n \in \mathcal{G}'_{\bm{C}}} \theta(n,n_1) \leq \arcsin r_c$
from Property \ref{prope2}. Let $\bar{m}\in \arg\min_{m \in \mathcal{G}_{\bm{P}}} \theta(m,m_1)$ and
$\bar{n}\in \arg\min_{n \in \mathcal{G}'_{\bm{C}}} \theta(n,n_1)$, in which the minimums are achieved due to the closedness of the system graphs. Then
inequality (\ref{ineq:angles1}) implies that
\begin{align*}
  \theta(m_1,n_1) &\geq \theta(\bar{m},\bar{n})-\theta(\bar{m},m_1)-\theta(\bar{n},n_1)\\
  &\geq \theta(\bar{m},\bar{n})-\arcsin r_p -\arcsin r_c\\
  &\geq \arcsin b_{{P},{C}}-\arcsin r_p -\arcsin r_c =: \epsilon > 0.
\end{align*}
Hence it holds that
$$\inf_{m_1 \in \mathcal{S}(\mathcal{G}_{{P}},r_p),n_1 \in \mathcal{S}(\mathcal{G}'_{{C}},r_c)} \arcsin \theta (m_1,n_1) \geq \epsilon > 0,$$
which implies $$\mathcal{S}(\mathcal{G}_{{P}},r_p)\cap \mathcal{S}(\mathcal{G}'_{{C}},r_c) = \{0\},$$ as required.
\end{proof}

A short summary to the above results follows. A certain robust stability condition is derived while allowing simultaneous perturbations on the plant
and controller, in the expression of an ``$\arcsin$'' inequality. The uncertainties are measured with conelike neighborhoods. It is worth noting that
for nonlinear systems, $\delta$-type gaps and $\gamma$-type gaps can be used to characterize the set of all neighboring system graphs within some
radius\cite{vinnicombe2000uncertainty}, which defines a set of manifolds. On the other hand, a conelike neighborhood simply gathers all input-output
pairs of a certain distance from the center, which forms a manifold itself. The advantage of focusing on input-output pairs instead of system graphs
arises in the case where only partial information about the graph of a nonlinear system is available, say in the form of some measured input-output
data set, which may not be sufficient for the purpose of computing the gap-distance, rendering standard gap-type stability conditions inapplicable. On
the contrary, if the uncertainties are measured with respect to the available input-output pairs, it is likely that the limited measured data are
sufficient to give a good approximation of these uncertainties. To verify whether a partially known perturbed system lies within a conelike
neighborhood, it suffices to check every available input-output pair.

\section{Networked Robust Stabilization with Cascaded Nonlinear Uncertainties}\label{sec:network}

\subsection{Two-Port Networks as Communication Channels}

\begin{figure}
    \centering
    \begin{subfigure}[b]{0.24\textwidth}
        \includegraphics[width=\textwidth]{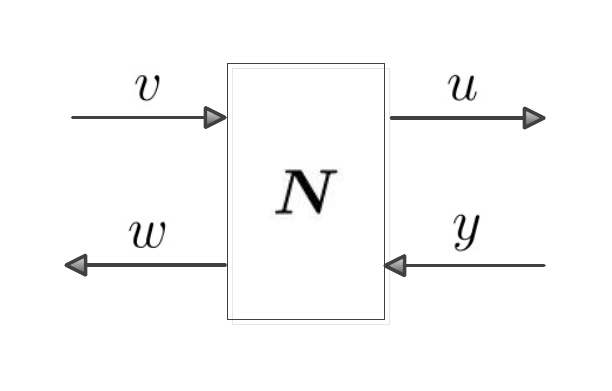}
        \caption{A Single Two-Port Network}
        \label{figtwoport}
    \end{subfigure}
    ~ 

    \begin{subfigure}[b]{0.34\textwidth}
        \includegraphics[width=\textwidth]{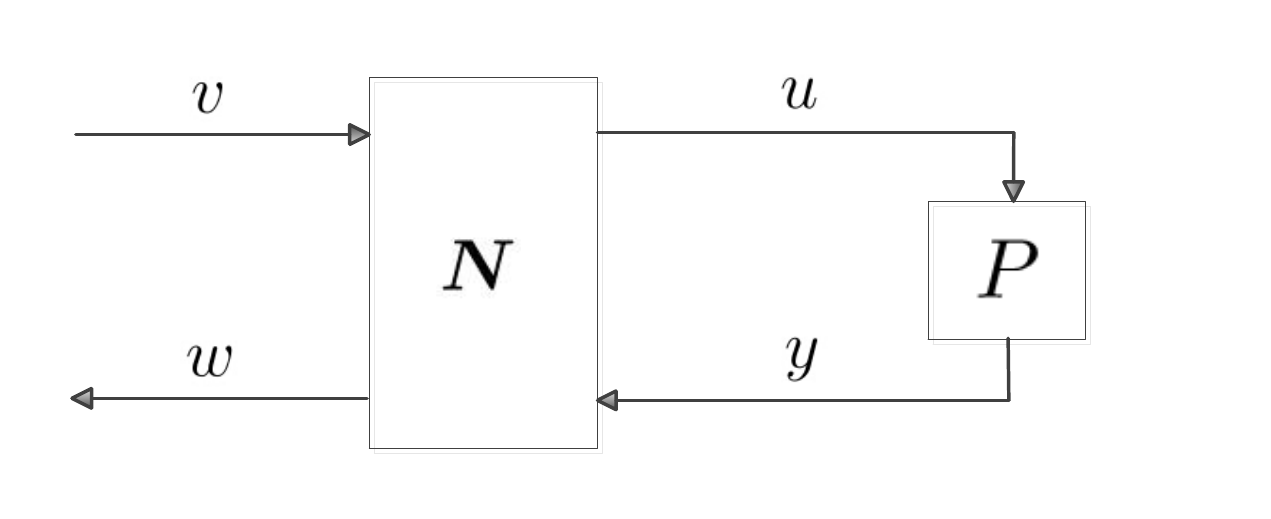}
        \caption{One-stage Two-Port Connection}
        \label{fig:oneblock}
    \end{subfigure}
    \caption{Two-Port Networks: an Illustration}\label{fig:two_port_of2}
    \vspace{-5pt}
\end{figure}

%
The use of two-port networks as a model of communication channels is adopted from \cite{gu2011cdc,di2016cdc}. \textcolor{black}{Two-port networks were
  first introduced and investigated in electrical circuits theory \cite{bakshi2009network}.} The network $\mathbd{N}$ in Fig.~\ref{figtwoport} has two
external ports, with one port composed of $v$, $w$ and the other of $u$, $y$, and is called a two-port network. A two-port network $\mathbd{N}$ may
have various representations, out of which we choose the transmission type to model a communication channel. Define the transmission matrix $\bmT$ as
\begin{align}\label{eq:transmission}
\bmT = \begin{bmatrix}
\bmT_{11} & \bmT_{12} \\
\bmT_{21} & \bmT_{22} \\
\end{bmatrix}~\text{and} ~\begin{bmatrix}
     v \\
     w \\
   \end{bmatrix} =  \bmT\begin{bmatrix}
     u \\
     y \\
   \end{bmatrix}.
\end{align}


When the communication channel is perfect, i.e., communication takes place without distortion or interference, the transmission matrix is simply
$$\bmT = \begin{bmatrix}
    \bmI_m & 0 \\
    0 & \bmI_p \\
    \end{bmatrix}.$$
If the bidirectional channel admits both distortions and interferences, we can let the transmission matrix take the form
$$\bmT = \bmI+\bmD= \begin{bmatrix}
\bmI_m+\bmD_\div & \bmD_- \\
\bmD_+ & \bmI_p+\bmD_\times \\
\end{bmatrix},$$
where $\bmI: \mathcal{H}_2 \mapsto \mathcal{H}_2$ is the identity operator and $$\bmD = \begin{bmatrix}
\bmD_\div & \bmD_- \\
\bmD_+ & \bmD_\times \\
\end{bmatrix}: \mathcal{H}_2 \mapsto \mathcal{H}_2$$
satisfies $\|\bmD\|\leq r<1$, which ensures that $\bmT$ is stably invertible. The four-block matrix $\bmD$ is called the uncertainty
quartet. 

\subsection{Graph Analysis on Cascaded Two-Port NCS}\label{subsec:NCS with Two-Port}

It is well known that graphs symbols can be defined for finite-dimensional LTI systems \cite{zhou1998essentials}. For every LTI system with transfer
function $P$, it admits a right coprime factorization $P = NM^{-1}$ satisfying $X M + Y N = I$, where $M,N,X,Y \in \mathcal{RH}_\infty$. The graph
symbol is defined as $$\begin{bmatrix}
  M \\
  N \\
\end{bmatrix}, ~\text{whereby}~ \mathcal{G}_P = \begin{bmatrix}
M \\
N \\
\end{bmatrix}\mathcal{H}_2;$$
see \cite[Proposition 1.33]{vinnicombe2000uncertainty}.

As illustrated in Fig.~\ref{fignetwork}, the LTI plant $P = NM^{-1}$ and LTI controller $C = VU^{-1}$ communicate with each other through cascaded
two-port networks involving nonlinear perturbations. In particular, one can characterize the input-output pairs in the graph
of $P$ as
$$\begin{bmatrix}
u \\
y \\
\end{bmatrix}
= \begin{bmatrix}
M \\
N \\
\end{bmatrix}x,$$ where $x \in \mathcal{H}_2$.

Consider the transmission type representation of the two-port networks $\{\mathbd{N}_k\}_{k=1}^{l}$.
If the $k$-th stage of the network admits a stable nonlinear uncertainty $\bmD_k$, then the transmission matrix is given as $\bmT_k =
\bmI+\bmD_k$. Signals in Fig.~\ref{fig:equiv_con_pla} have the following relations:
\begin{align*}
  \begin{bmatrix}
    u_k \\
    y_k \\
  \end{bmatrix} = \left(\prod _{j=1}^{k} \bmT_{k+1-j}\right)\begin{bmatrix}
    u \\
    y \\
  \end{bmatrix} = \left(\prod _{j=1}^{k}(\bmI+\bmD_{k+1-j})\right)
  \begin{bmatrix}
    u \\
    y \\
  \end{bmatrix},
\end{align*}
\begin{align*}
  \begin{bmatrix}
    v_{k} \\
    w_{k} \\
  \end{bmatrix} &= \left(\prod _{j=k+1}^{l} \bmT_j^{-1}\right) \begin{bmatrix}
    v \\
    w \\
  \end{bmatrix}= \left(\prod _{j=k+1}^{l}(\bmI+\bmD_{j})^{-1}\right)\begin{bmatrix}
    v \\
    w \\
  \end{bmatrix}.
\end{align*}

\begin{figure}
  \centering
  \includegraphics[scale=0.45]{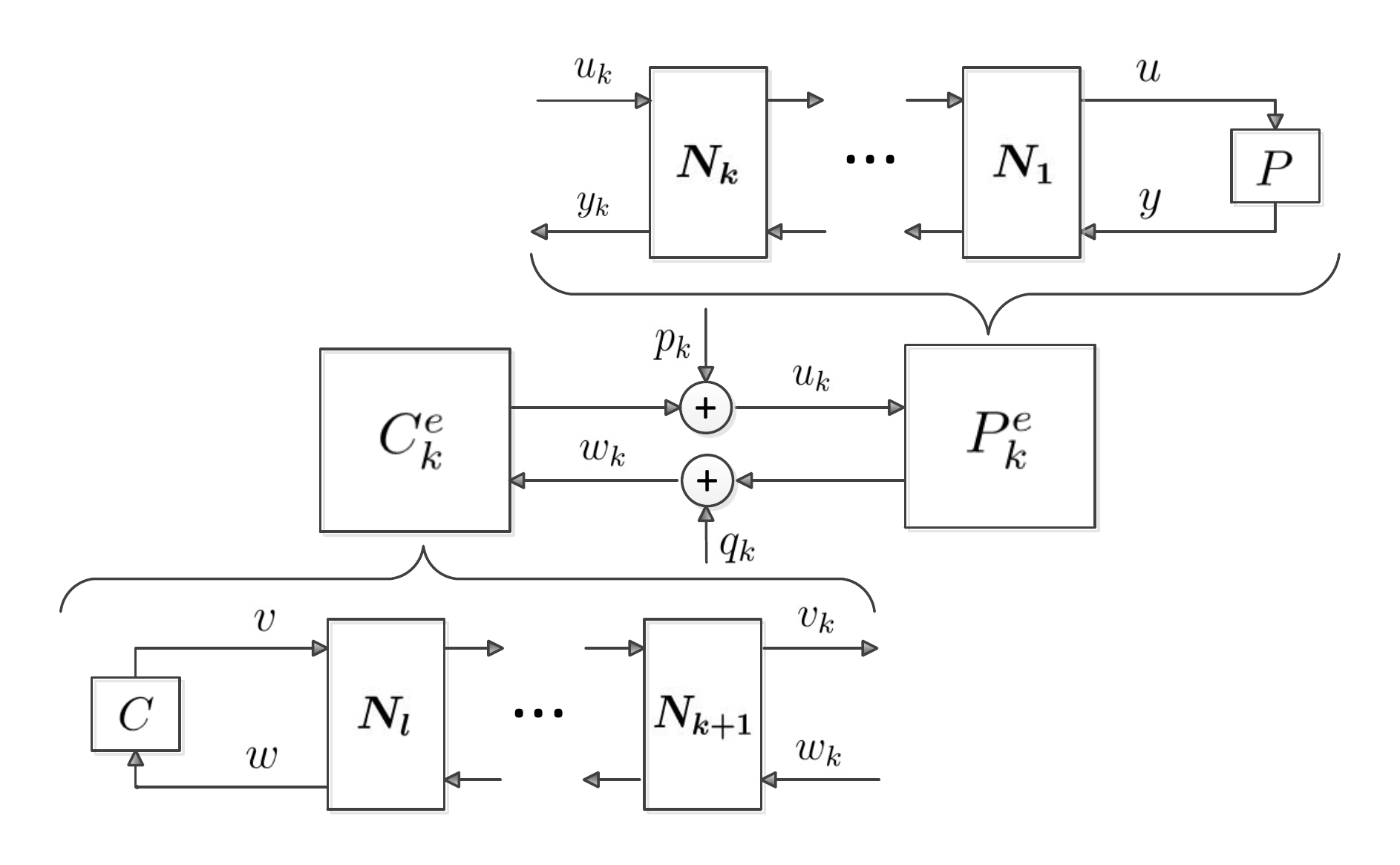}\\
  \caption{Equivalent Plant and Controller}\label{fig:equiv_con_pla}
  \vspace{-5pt}
\end{figure}

If we view $P$ together with $\{\mathbd{N}_j\}_{j=1}^{k}$ as an equivalent plant $\bmP^{e}_k$ with uncertainties $\{\bmD_j\}_{j=1}^{k}$, then the
graph of $\bmP^{e}_k$ is given by
\begin{equation}\label{eq:G_Pk}
\begin{aligned}
\mathcal{G}_{\bm{P}^{e}_k} = \left(\prod _{j=1}^{k}(\bmI+\bmD_{k+1-j})\right)\mathcal{G}_{P}.
\end{aligned}
\end{equation}

Similarly, if we view $C$ together with $\{\mathbd{N}_j\}_{j=k+1}^{l}$ as an equivalent controller $\bmC^{e}_k$ with uncertainties
$\{\bmD_j\}_{j=k+1}^{l}$, then the graph of $\bmC^{e}_k$ is
\begin{equation}\label{eq:G_Ck}
\begin{aligned}
\mathcal{G}'_{\bm{C}^{e}_k} = \left(\prod _{j=k+1}^{l}(\bmI+\bmD_{j})^{-1}\right)\mathcal{G}'_{C}.
\end{aligned}
\end{equation}

For convenience, we regard $k = 0$ as the situation when $P$ is isolated from the two-port networks and $k = l$ when $C$ is isolated.

%

\subsection{Robust Stability Condition}
With the equivalent plant and controller representations derived aforehand, next we extend the definition on the stability of the two-port NCS in
\cite{di2016cdc} to the nonlinear case.

As shown in Fig.~\ref{fig:equiv_con_pla}, we denote the $k$-th input pair as $I_k := [p_k,q_k]^T$, the $k$-th output pair as $O_k:= [u_k,w_k]^T$ and
the set of all outputs as $O := [u_1,w_1,u_2,w_2,\dots,u_l,w_l]^T$. By the feedback well-posedness assumption, the map from input $I_k$ to output $O$
exists and we denote it as $\bm{A}_k: I_k \in \mathcal{H}_2 \mapsto O \in \mathcal{H}_2$.
\begin{definition}
The two-port NCS in Fig.~\ref{fig:equiv_con_pla} is said to be stable if the operator $\bm{A}_k$ is finite-gain stable for every $k=0,1,\dots,l$.
\end{definition}
%
The following proposition further simplifies the stability condition.

\begin{proposition}\label{prop:stability}
The two-port NCS is finite-gain stable if and only if the equivalent closed-loop system $[\bmP^{e}_k,\bmC^{e}_k]$ is finite-gain stable for every $k = 0,1,...,l$.
\end{proposition}

\begin{proof}
Necessity holds trivially. Below we show sufficiency.

Let $[\bmP^{e}_k,\bmC^{e}_k]$ be finite-gain stable, and thus $\bmF^{-1}_{\bm{P^{e}_k},\bm{C^{e}_k}}$ is stable. As $\|\bmD\| < 1$ by hypothesis,
both $\bmI+\bmD_j$ and $(\bmI+\bmD_j)^{-1}$ are stable. Hence the composite map of $\bmF^{-1}_{\bm{P^{e}_k},\bm{C^{e}_k}}$ and $\bmI+\bmD_j$ or
$(\bmI+\bmD_j)^{-1}$ is stable, which implies the stability of $\bm{A}_k$ for all $k = 0, 1, \ldots l$.
\end{proof}
With the stability definition at hand, we present next the main robust stability  theorem involving nonlinear perturbations in a two-port NCS.

In the following we assume that every closed-loop system $[\bmP,\bmC]$ is well-posed and $\bmF_{\bm{P},\bm{C}}$ is surjective.
Hence from Proposition \ref{prop:pi_pc_stability}, the stability of $[\bmP,\bmC]$ is equivalent to the finite-gain stability of $\Pi_{\mathcal{G}_{\bm{P}}\sslash\mathcal{G}'_{\bm{C}}}$.
Let nominal LTI closed-loop system $[P,C]$ be stable.
\begin{theorem}\label{thm:main}
\textcolor{black}{The two-port NCS is finite-gain stable for all $\{\bm{\Delta}_k\}_{k=1}^{l}$ subject to $\|\bm{\Delta}_k\| \leq r_k$ if and only if}
\begin{align}\label{eq:two-port_arcsin}
\sum _{k=1}^{l} \arcsin r_k < \arcsin b_{P,C}.
\end{align}

\end{theorem}
From the above theorem, we know the stability margin $b_{P,C}$ is the same as that in a standard closed-loop system with ``gap'' uncertainties \cite{zhou1998essentials,qiu1992feedback, qiu1992pointwise}, hence the synthesis problem of a two-port NCS can be solved by an $\mathcal{H}_\infty$ optimization. In addition, the synthesis is irrelevant to detailed requirements of communication channels between the plant and controller, such as the number of two-port connections and how the uncertainty bounds are distributed among all the channels, which provide more flexibility on the selection of the communication channels.

Before proceeding to the proof of Theorem \ref{thm:main}, we introduce a useful lemma.

\begin{lemma}\label{lem:cascaded}
Given $r_1,r_2 \in (0,1)$ and \textcolor{black}{a closed conelike neighborhood} $\mathcal{M} \subset \mathcal{H}_2$, it holds that
$$\mathcal{S}(\mathcal{S}(\mathcal{M},r_1),r_2)\subset\mathcal{S}(\mathcal{M},\sin(\arcsin r_1 + \arcsin r_2)).$$
\end{lemma}
\begin{proof}
Let $\mathcal{M}_1 = \mathcal{S}(\mathcal{M},r_1)$. Let $m \in \mathcal{M}$, $m_1 \in \mathcal{M}_1$ and $m_2 \in \mathcal{S}(\mathcal{S}(\mathcal{M},r_1),r_2)$. Then we have
\begin{align}\label{ineq:angles2}
  \theta(m_2,m) \leq \theta(m_2,m_1) + \theta(m_1,m).
\end{align}
Particularly, take $\bar{m}_1 \in \arg\min_{m_1 \in \mathcal{M}_1} \theta(m_2,m_1)$. Then inequality (\ref{ineq:angles2}) implies that
\begin{align*}
  \theta(m_2,m) &\leq \theta(m_2,\bar{m}_1) + \theta(\bar{m}_1,m) \\
  &\leq \arcsin r_2 + \theta(\bar{m}_1,m).
\end{align*}
Taking infimum at the both sides brings about that
\begin{align*}
\inf_{m \in \mathcal{M}}\theta(m_2,m) &\leq \arcsin r_2 + \inf_{m \in \mathcal{M}}\theta(\bar{m}_1,m) \\
&\leq \arcsin r_1 + \arcsin r_2.
\end{align*}

Hence, $m_2 \in \mathcal{S}(\mathcal{M},\sin(\arcsin r_1 + \arcsin r_2)),$ which completes the proof.
\end{proof}
\textcolor{black}{The above lemma characterizes the inclusion relations of conelike sets.} The proof of Theorem
\ref{thm:main} is given next.
%
%

\begin{proof}[Proof of Theorem~\ref{thm:main}]
\textcolor{black}{The necessity follows from the ``arcsin'' theorem in \cite{di2016cdc} for LTI systems by
  noting that that the linear two-port neighborhood $\mathcal{N}(\mathcal{G}_{{P}},r)$ is contained in the conelike set
  $\mathcal{S}(\mathcal{G}_{{P}},r)$.}

Next we prove the sufficiency.
Assume we are at the $k$-th stage of equivalent closed-loop system as shown in Fig.~\ref{fig:equiv_con_pla}. Let $\mathcal{M} = \mathcal{G}_{\bm{P}}$ and $\mathcal{M}^e_j = \mathcal{G}_{P^e_j}$, $j=1,2,\dots,l$. Then $$\mathcal{M}^e_{j}=\bmT\mathcal{M}^e_{j-1} = (\bmI+\bmD_j)\mathcal{M}^e_{j-1}$$ with $\|\bmD_j\|\leq r_j$. Let $ n \in \mathcal{M}^e_{j} \setminus \{0\}$, there exists an $m_1 \in \mathcal{M}^e_{j-1}$ such that $n=(\bmI+\bmD_j)m_1$. Hence we have
$$\inf_{0\neq m \in \mathcal{M}^e_{j-1}}\frac{\|n-m\|_2}{\|m\|_2} \leq \frac{\|\bmD_j m_1\|_2}{\|m_1\|_2} \leq \|\bmD_j\|\leq r_j.$$
As a result,
$\mathcal{M}_j^e \subset \mathcal{S}(\mathcal{M}^e_{j-1},r_j)$, $j=2,3,\dots,k$.

From Lemma \ref{lem:cascaded} and by induction, we have
$$\mathcal{M}^e_k \subset \mathcal{S}(\mathcal{M}^e_{k-1},r_k) \subset \dots \subset \mathcal{S}(\mathcal{M},\sin(\sum_{j=1}^{k} \arcsin r_j)).$$

\textcolor{black}{Likewise, for the controller part, let $\mathcal{N} = \mathcal{G}'_{\bm{C}}$ and $\mathcal{N}^e_j =
  \mathcal{G}'_{C^e_j}$.
  Then $$\mathcal{N}^e_{j-1} = \bmT^{-1}\mathcal{N}^e_{j} = (\bmI+\bmD_j)^{-1}\mathcal{N}^e_{j}$$ with $\|\bmD_j\|\leq r_j$. Given any
  $ n \in \mathcal{N}^e_{j-1}\setminus \{0\}$, there exists an $m_1 \in \mathcal{N}^e_{j}$ such that $n=(\bmI+\bmD_j)^{-1}m_1$. Hence we have}
$$\textcolor{black}{\inf_{0\neq m \in \mathcal{N}^e_{j}}\frac{\|n-m\|_2}{\|n\|_2} \leq \frac{\|\bmD_j n\|_2}{\|n\|_2} \leq \|\bmD_j\|\leq r_j.}$$
\textcolor{black}{By Property \ref{prope3}, we have $$\mathcal{N}^e_{j-1} \subset \tilde{\mathcal{S}}(\mathcal{N}^e_{j},r_j) = {\mathcal{S}}(\mathcal{N}^e_{j},r_j), ~j = k+1,\dots,l.$$
Hence by the same arguments as above, we have}
$$\mathcal{N}^e_k \subset \mathcal{S}(\mathcal{N},\sin( \sum_{j=k+1}^{l} \arcsin r_j)).$$

Therefore, from the stability condition (\ref{eq:two-port_arcsin}) and Theorem \ref{thm:extend_arcsin}, we know $[\bmP_k^e, \bmC_k^e]$ is stable for
every $k = 0,1,...,l$. Combining this with Proposition \ref{prop:stability}, we obtain the finite-gain stability of the two-port NCS.
\end{proof}
%
%
\subsection{Scalability of the Stability Condition}
When we are faced with a large-scale network with many relays and connections, a particular communication link between a plant and a controller may
involve many cascaded two-port channels. As the topology of an NCS changes, we need to confirm whether the new two-port communication link is
``healthy'' enough to keep the NCS robustly stable. Revaluating the whole network from the beginning may be impractical due to the limitations on
computational resources or responding time. In the following, we show this problem can be solved in the two-port NCS by defining the stability
residue properly.

For a two-port NCS with an LTI plant $P$ and LTI controller $C$ under nonlinear perturbations on its communication channels, define its
stability \textcolor{black}{residue} as
\begin{align}\label{eq:Rpc}
R_{P,C}(r_1,\dots,r_l):=\arcsin b_{P,C} - \sum_{k=1}^l \arcsin r_k,
\end{align}
which is subsequently written as $R_{P,C}$ without ambiguity.
It follows from Theorem \ref{thm:main} that the two-port NCS is stable for all stable uncertainties $\{\bm{\Delta}_k\}_{k=1}^{l}$ subject to $\|\bmD_k\| \leq r_k$ if and only if $R_{P,C}(r_1,\dots,r_l)>0$. It is no doubt that the larger $R_{P,C}$ is, the more robustly stable the NCS will be.

When some new two-port connections are added or some old ones are modified, checking whether the resulting NCS remains robustly stable becomes
necessary. For this purpose, one only needs to update the stability residue and check its feasibility.
\begin{itemize}
  \item When a new connection $\bmT_{\text{new}} = \bmI+\bmD_{\text{new}}$ satisfying $\|\bmD_{\text{new}}\|\leq r_{\text{new}}$ is added, let $$R^{\text{new}}_{P,C}\leftarrow R_{P,C}-\arcsin r_{\text{new}}.$$
  \item When an old connection $\bmT_{\text{old}} = \bmI+\bmD_{\text{old}}$ with $\|\bmD_{\text{old}}\|\leq r_{\text{old}}$ is changed to
    $\bmT_{\text{new}} = \bmI+\bmD_{\text{new}}$ satisfying $\|\bmD_{\text{new}}\|\leq r_{\text{new}}$, let
      $$R^{\text{new}}_{P,C}\leftarrow R_{P,C} - \arcsin r_{\text{new}}+\arcsin r_{\text{old}}.$$
\end{itemize}
It follows from Theorem \ref{thm:main} and Equation (\ref{eq:Rpc}) that the new NCS will be robustly stable if and only if $R^{\text{new}}_{P,C}>0$ after
sequentially updating $R_{P,C}$ with respect to all the changes. In other words, the stability condition given in Theorem \ref{thm:extend_arcsin} is
scalable as the network size is enlarged.

\section{Conclusion}\label{sec:conc}
We investigate networked robust stabilization problem concerning LTI systems perturbed by nonlinear uncertainties.
A special conelike uncertainty set is studied, which bridges the techniques of handling linear subspaces to those of handling nonlinear uncertainties
in cascaded two-port networks. \textcolor{black}{A necessary and sufficient stability condition is given in the form of an ``arcsin'' inequality,
  which is scalable when the size of the network is enlarged. As far as control synthesis is concerned, the problem can be solved through an
  $\mathcal{H}_\infty$ optimization of the closed-loop stability margin.}

\appendices
\section{Proofs of Properties \ref{prope1},~\ref{prope2} and \ref{prope3}}\label{app:pf_properties}
\begin{proof}
  \textcolor{black}{From the closedness of the conelike set $\mathcal{M} \subset \mathcal{H}_2$, we can replace ``inf'' with ``min'' in the
    definition of $\tilde{\mathcal{S}}(\mathcal{M},r)$ and ${\mathcal{S}}(\mathcal{M},r)$. Next we prove the properties in turns.}

\textcolor{black}{For Property \ref{prope1}, it suffices to show for $n \in \mathcal{S}(\mathcal{M},r)$, it holds $\alpha n \in \mathcal{S}(\mathcal{M},r)$ for every $\alpha \in \mathbb{R}$.
Using the definition, we have}
$$\inf_{0\neq m \in \mathcal{M}}\frac{\|\alpha n-m\|}{\|m\|} = \inf_{0\neq \alpha m \in \mathcal{M}}\frac{\|\alpha n-\alpha m\|}{\|\alpha m\|} \leq r,$$
which establishes Property \ref{prope1}.

For Property \ref{prope2}, let $n \in \mathcal{S}(\mathcal{M},r)$. It follows that
$$\bar{m} \in \argmin_{ m\in \mathcal{M}} \frac{\|n-m\|_2}{\|m\|_2}$$
\textcolor{black}{satisfies $n-\bar{m} \perp n$, whereby $\sin \theta(\bar{m},n) = \frac{\|n-\bar{m}\|_2}{\|\bar{m}\|_2} \leq r$.  Consequently,
  $\mathcal{S}(\mathcal{M},r) \subset \{n\in \mathcal{H}_2 : \min_{m\in \mathcal{M}} \theta(m,n)\leq \arcsin r\}\cup \{0\}$. On the other hand, let
  $n$ belongs to the latter set. From Property \ref{prope1}, we can find $\bar{m} \in \mathcal{S}(\mathcal{M},r)$ such that
  $\theta(\bar{m},n) = \arcsin r $ and $n-\bar{m} \perp n$, which implies that }
$$\frac{\|n-\bar{m}\|_2}{\|\bar{m}\|_2} = \sin \theta(\bar{m},n)\leq r.$$
This completes the proof for Property \ref{prope2}.

\textcolor{black}{For Property \ref{prope3}, given $n \in \tilde{\mathcal{S}}(\mathcal{M},r)$, consider}
$$\bar{m} \in \arg\min_{ m\in \mathcal{M}} \frac{\|n-m\|_2}{\|n\|_2}.$$
\textcolor{black}{It follows that $n-\bar{m} \perp \bar{m}$. Denote the acute angle between $n$ and $\bar{m}$ as $\theta_0 (\leq \arcsin r)$. In the hyperplane determined by $\bar{m}$ and $n$, as shown in Fig.~
\ref{fig:conelike}, we can extend $\bar{m}$ to $\bar{m}_1 \in \mathcal{M}$ along $\bar{m}$ such that $n-\bar{m}_1 \perp n$, which is guaranteed by Property \ref{prope1}. This implies that} $$\min_{0 \neq m\in \mathcal{M}} \frac{\|n-m\|_2}{\|m\|_2} \leq \frac{\|n-\bar{m}_1\|_2}{\|\bar{m}_1\|_2} = \sin \theta_0 \leq r.$$
Consequently, $\tilde{\mathcal{S}}(\mathcal{M},r)\subset {\mathcal{S}}(\mathcal{M},r)$.

\textcolor{black}{On the other hand, given $n \in {\mathcal{S}}(\mathcal{M},r)$ and consider}
$$\bar{m} \in \arg\min_{0 \neq m\in \mathcal{M}} \frac{\|n-m\|_2}{\|m\|_2},$$
\textcolor{black}{one can argue likewise that ${\mathcal{S}}(\mathcal{M},r)\subset \tilde{\mathcal{S}}(\mathcal{M},r)$, which completes the proof.}
\end{proof}

\begin{figure}
  \centering
  \includegraphics[width=.35\textwidth]{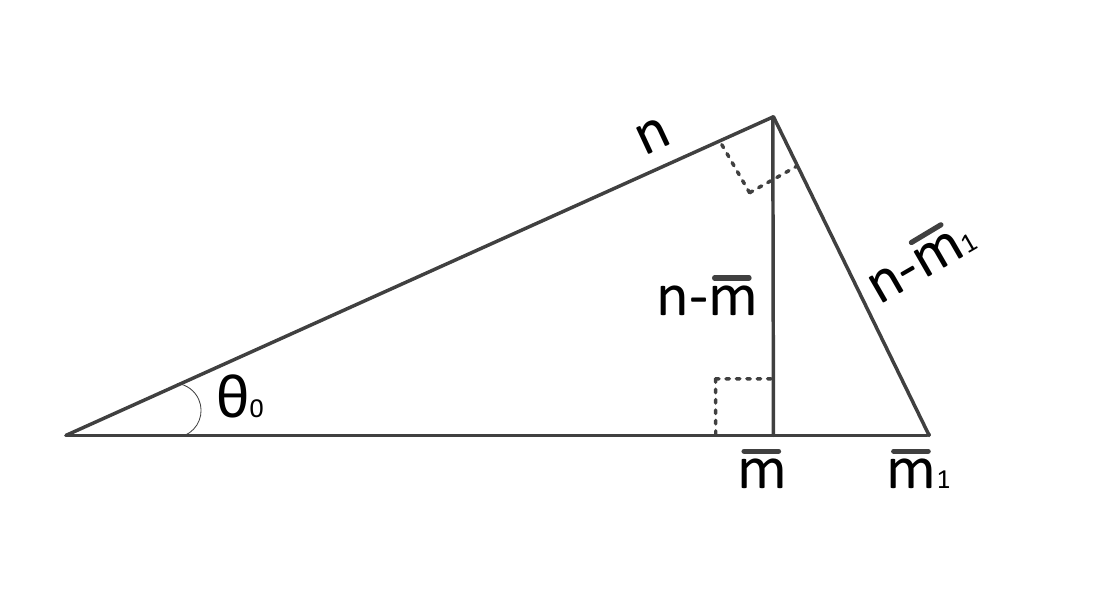}\\
  \caption{Illustration for Proof of Property \ref{prope3}}\label{fig:conelike}
\end{figure}

\section{Proof of Proposition \ref{prop:nonlinear_stability_criterion}}\label{app:pf_propo2}
\begin{proof}
  We prove both parts by contradiction. The proposition holds trivially when $[P,C]$ is unstable, thus it suffices to prove the case when $[P,C]$ is
  stable. For brevity, define $\mathcal{M}:=\mathcal{G}_{{P}}$, $\mathcal{N}:=\mathcal{G}'_{{C}}$, $\mathcal{M}_1:=\mathcal{G}_{\bm{P}_1}$ and
  $\mathcal{N}_1:=\mathcal{G}'_{\bm{C}_1}$.

\textbf{Sufficiency:}

Let $\mathcal{G}_{\bm{P}_1}\subset \mathcal{S}(\mathcal{G}_{{P}},r_p),~\mathcal{G}'_{\bm{C}_1}\subset \mathcal{S}(\mathcal{G}'_{{C}},r_c)$ and
$[\bmP_1,\bmC_1]$ be unstable. It follows that $\Pi_{\mathcal{M}_1\sslash\mathcal{N}_1}$ is unbounded. That is to say, there exists a sequence
$\{\omega_k\}_{k=1}^{\infty} \subset \mathcal{H}_2\setminus \{0\}$, such that:
\begin{itemize}
  \item $\|\omega_k\|_2 \nearrow \infty$;
  \item $\displaystyle\lim _{k \rightarrow \infty}\frac{\|\Pi_{\mathcal{M}_1\sslash\mathcal{N}_1} \omega_k\|_2}{\|\omega_k\|_2} = \infty$.
\end{itemize}

By the surjectivity of $\bmF_{\bm{P}_1,\bm{C}_1}$, we know that
$$\omega_k = \Pi_{\mathcal{M}_1\sslash\mathcal{N}_1}\omega_k + \Pi_{\mathcal{N}_1\sslash\mathcal{M}_1}\omega_k=: m_k+n_k.$$ Hence,
$\alpha_k := \frac{\|\omega_k\|_2}{\|m_k\|_2} \rightarrow 0$ as $k \rightarrow \infty$.  From Definition \ref{def:wellposed}, we know
$m_k,n_k \in \mathcal{H}_2 \setminus \{0\}$, and thus the angle between them can be computed as
$$\theta(m_k,n_k)=\arccos \left|\frac{\langle m_k,n_k\rangle}{\|m_k\|_2\|n_k\|_2}\right|.$$

Consequently,
\begin{align*}
\cos \theta(m_k,n_k)& = \left|\frac{\langle m_k,n_k\rangle}{\|m_k\|_2\|n_k\|_2}\right| \\
& \geq \left|\frac{\langle m_k,n_k\rangle}{\|m_k\|_2(\|m_k\|_2+\|\omega_k\|_2)}\right| \\
& = \frac{1}{1+\alpha_k}\left(1 - \frac{|\langle m_k,\omega_k\rangle|}{\|m_k\|^2_2}\right)\\
& \geq \frac{1}{1+\alpha_k}\left(1-\frac{\|m_k\|_2\|\omega_k\|_2}{\|m_k\|^2_2}\right)\\
& = \frac{1-\alpha_k}{1+\alpha_k} \rightarrow 1 ~~~~\text{as} ~k \rightarrow \infty.
\end{align*}

Hence $\theta(m_k,n_k) \rightarrow 0$. Since $\mathcal{G}_{\bm{P}_1}$ and $\mathcal{G}'_{\bm{C}_1}$ are closed sets, it follows that
$\mathcal{G}_{\bm{P}_1} \cap \mathcal{G}'_{\bm{C}_1} \neq \{0\}$ and therefore
$\mathcal{S}(\mathcal{M},r_p)\cap \mathcal{S}(\mathcal{N},r_c) \neq \{0\}$, which leads to a contradiction.

\textbf{Necessity:}

Assume there exists a nonzero $u$ satisfying $u \in \mathcal{S}(\mathcal{M},r_p) \cap \mathcal{S}(\mathcal{N},r_c)$. From Property \ref{prope1} we know $\{\alpha u:\alpha \in \mathbb{R}\} \subset \mathcal{S}(\mathcal{M},r_p)\cap \mathcal{S}(\mathcal{N},r_c)$.
Construct two scalar sequences $\{\alpha_k\}_{k=1}^{\infty}$ and $\{\beta_k\}_{k=1}^{\infty} \subset \mathbb{R}$ such that
\begin{itemize}
  \item $|\alpha_k|,~|\beta_k|$ and $|\frac{\alpha_k}{\beta_k}| \nearrow \infty$;
  \item $\beta_k = \alpha_t+(\beta_l-\alpha_l)$ if and only if $k=t=l$.
\end{itemize}

Furthermore, construct two graphs
$\mathcal{M}_1:=\mathcal{G}_{\bm{P}_1}\subset \mathcal{S}(\mathcal{M},r_p),~\mathcal{N}_1:=\mathcal{G}'_{\bm{C}_1}\subset
\mathcal{S}(\mathcal{N},r_c)$,
such that $\{\alpha_k u\} \subset \mathcal{M}_1$, $\{(\beta_k-\alpha_k) u\} \subset \mathcal{N}_1$ and $\bmF_{\bm{P}_1,\bm{C}_1}$ is
surjective. Hence, for any $\omega_k = \beta_k u \in \mathcal{H}_2$, we have the decomposition
$$\omega_k = \beta_k u = \alpha_k u + (\beta_k-\alpha_k) u := m_k+n_k.$$
Moreover, $$\lim_{k \rightarrow \infty} \frac{\|m_k\|_2}{\|\omega_k\|_2} =  \lim_{k \rightarrow \infty}|\frac{\alpha_k}{\beta_k}| = \infty.$$
It follows directly that $\Pi_{\mathcal{M}_1\sslash\mathcal{N}_1}$ is unbounded, i.e. $[\bmP_1,\bmC_1]$ is unstable, which leads to a contradiction.
\end{proof}
\bibliographystyle{IEEEtran}

\bibliography{mybibfile_cdc2017}
\end{document}